\documentclass[12pt,reqno]{amsart}

\usepackage{setspace}
\usepackage{amsmath}
\usepackage{amsfonts}
\usepackage{amssymb}
\usepackage{latexsym}   
\usepackage{graphicx}
\usepackage{color}
\usepackage{subfig}
\usepackage[table]{xcolor}
\usepackage{algorithm,algorithmic}
\usepackage[colorlinks]{hyperref}

\usepackage{vmargin,fancyhdr}   

\newcommand{\spara}[1]{\smallskip\noindent{\bf #1}}

\newcommand{\beql}[1]{\begin{equation}\label{#1}}
\newcommand{\eeq}{\end{equation}}

\newtheorem{theorem}{Theorem}
\newtheorem{definition}{Definition}

\setpapersize{USletter}
\setmarginsrb{.75in}{.5in}        
             {.75in}{.5in}        
             {.25in}{.25in}     
             {.25in}{.5in}      
\newcommand{\field}[1]{\mathbb{#1}} 

\newcommand{\beq}[1]{\begin{equation}\label{#1}}

\newcounter{rot}

\newcommand{\ignore}[1]{}

\title{Towards Quantifying Vertex Similarity in Networks}

\begin{document}

\author{Charalampos E. Tsourakakis}
\address{Department of Mathematical Sciences\\
Carnegie Mellon University\\
5000 Forbes Av., 15213\\
Pittsburgh, PA \\
U.S.A} \email{ctsourak@math.cmu.edu}

\date{\today}

\begin{abstract}
Vertex similarity is a major problem in network science with a wide range of applications. In this work we provide novel perspectives on finding 
(dis)similar vertices  within a network and across two networks with the same number of vertices (graph matching). 
With respect to the former problem, we propose to optimize a geometric objective which allows us to express each vertex uniquely 
as a convex combination of a few extreme types of vertices. Our method has the important advantage of supporting efficiently several 
types of queries such as ``which other vertices are most similar to this vertex?'' by the use of the appropriate data structures and 
of mining interesting patterns in the network. With respect to the latter problem (graph matching), we propose the generalized condition 
number --a quantity widely used in numerical analysis--  $\kappa(L_G,L_H)$ of the Laplacian matrix representations of
$G,H$ as a measure of graph similarity,  where $G,H$ are the graphs of interest.
We show that this objective has a solid theoretical basis and propose a deterministic and a randomized
graph alignment algorithm. 
We evaluate our algorithms on both synthetic and real data. We observe that 
our proposed methods achieve high-quality results and provide us with 
significant insights into the network structure.  

\end{abstract}

\maketitle
\section{Introduction}
\label{sec:intro}
Vertex similarity is an important network concept with a broad range of significant applications. 
Paradoxically, a major step towards the successful quantification of vertex similarity is finding a good definition of it. 
Typically, one is interested either in finding similar vertices in a given network 
or finding similar vertices across two different networks. 
The former problem emerges in numerous applications in social networks such as 
link prediction and recommendation. 
It also emerges in the domain of privacy since the improved ability 
of predicting an edge may be used for malicious purposes as well \cite{hay}. 
The latter problem emerges also in various domains including graph mining, computer vision and chemistry. 
The interested reader is urged to read \cite{zager} which contains a wealth of applications. 
In this work we provide novel perspectives on the two aforementioned problems.
On purpose we state them abstractly using quotes in several places, in order to emphasize that two main contributions of our work are two novel formalizations of these problems. 
The first problem is: {\em given an undirected graph $G(V,E)$ and two vertices $u,v \in V$, how ``similar'' are $u$ and $v$?}
The second problem one is: {\em given two graphs $G(V_G,E_G)$, $H(V_H,E_H)$ such that $|V_G|=|V_H|$ is there a permutation of the vertices of $H$ 
that ``reveals any similarities'' between $G$ and $H$? Can we find such a permutation efficiently?}
We will refer to these problems as the {\it vertex similarity} and the {\it graph matching} problem respectively.

\begin{figure*}
  \centering
  \subfloat[Vertex Similarity]{\label{fig:fig1a}\includegraphics[scale=0.4]{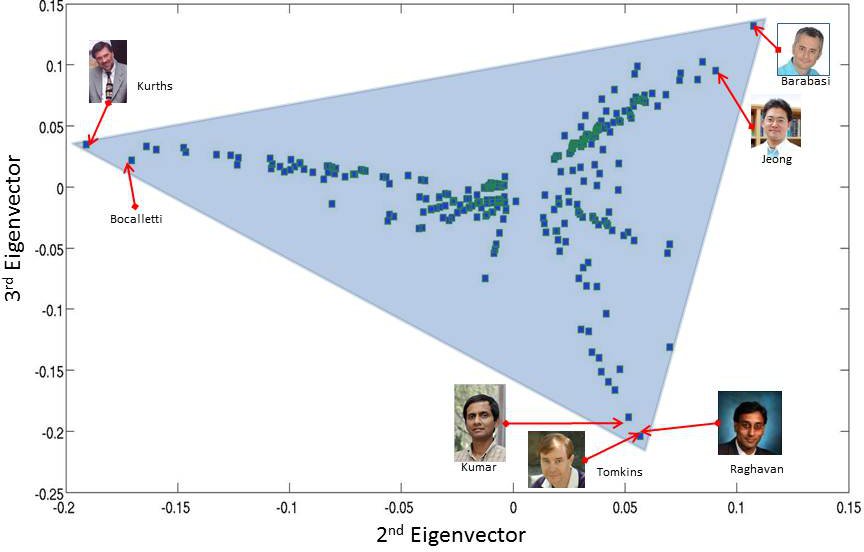}}
  \subfloat[Graph Matching]{\label{fig:fig1b}\includegraphics[scale=0.35]{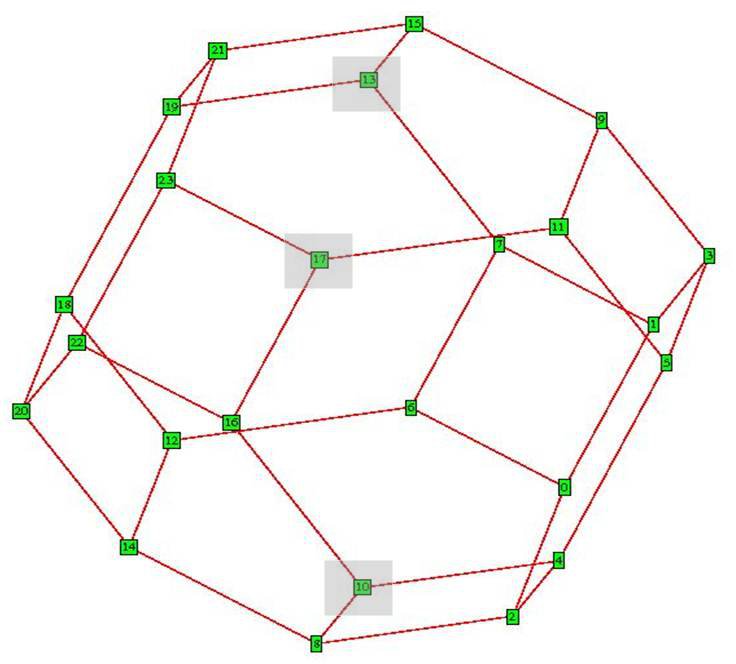}}
  \caption{ (a) Minimum area 2-simplex $\mathcal{S}$ for an informative embedding of the largest component of the Netscience network $G$, see Table~\ref{tab:datasets} for the  dataset details. $\mathcal{S}$ allows us to express each data point as a unique convex combination of its extreme points and hence  call two vertices of $G$ similar if their corresponding mixture coefficients are close. (b) 
   Permutahedron $\mathcal{P}$ of the symmetric group $S_4$. The 3 shaded vertices of $\mathcal{P}$ define a hypothetical set of isomorphisms between $G$ and  $H$. }
  \label{fig:fig1}
\end{figure*}

\spara{Paper contributions and roadmap.}
Our contributions are summarized as follows.

\begin{itemize}
 \item For the vertex similarity problem: 
 \begin{itemize}
 \item We propose a novel approach which is inspired by the concept of archetypal analysis \cite{cutler2}.
 We formalize our problem as optimizing a geometric objective, 
 namely finding an enclosing simplex of minimum volume that is robust to outliers for a special cloud of points.
 \item We propose an efficient algorithm for optimizing our objective.
 An output example of our method is shown in Figure~\ref{fig:fig1}(a) which 
 shows a minimum area 2-simplex $\mathcal{S}$ for a normalized Laplacian graph 
 embedding \cite{belkin} of the largest component of the Netscience network, see Table~\ref{tab:datasets}. 
 Using the Euclidean distance between the mixture coefficients (the smaller the distance is, the more similar the vertices are) we find  
 that the vertices `Prabhakar Raghavan', `Ravi Kumar' and `Andrew Tomkins' are highly similar
 and that the  three extreme points (archetypes) of  $\mathcal{S}$ correspond to three influential groups of researchers. 
 Specifically, the three vertices of the simplex lie close to Kurths and Bocalletti, Barabasi and Jeong, Kumar, Raghavan, Tomkins and Rajagopalan
 which are respectively three authoritative groups of researchers on social networks. 
 \item Our method has the advantage of supporting {\em efficiently} queries of the type ``which other vertices are most similar to this vertex?'', 
 ``which are the most dissimilar vertices to this vertex?'', by the use of the appropriate data structures  \cite{mount}.
\end{itemize}
\item For the graph matching problem:
 \begin{itemize}

\item We introduce a novel criterion of similarity between two graphs $G,H$, 
the {\em generalized condition number} of their Laplacian matrix representations. 

\item In constrast to frequently used heuristics, our criterion has a solid theoretical basis. 
Specifically, consider Figure~\ref{fig:fig1}(b) which shows the permutahedron $\mathcal{P}$ of the symmetric group $S_4$ \cite{burk}
with three shaded vertices which correspond to three hypothetical permutations which make graph $H$ identical to graph $G$. 
Theorem~\ref{thrm:thrm1}  proves  that the global minimum of our proposed objective function
occurs at the permutations that make $H$ equal to $G$. 
Specifically, our proposed Metropolis chain
in the limit of $\lambda \rightarrow +\infty$ ($\lambda \geq 1$ is a parameter in our chain)  converges to the uniform distribution over 
the shaded set, i.e., $\pi\Big({\text{Shaded Vertex}}\Big) = 1/3$, $\pi\Big( {\text{Non-Shaded Vertex}} \Big) = 0$.

\item Despite the fact that in this work we focus on the restricted version of the graph matching problem where 
$G,H$ have the same number of vertices, this does not make our work useless for two main reasons. 
First, there exist applications where $|V_G|=|V_H|$ \cite{zager}. Secondly, 
our conceptual contribution is likely to be extendable to the more general case where $|V_G| \neq |V_H|$, see Section~\ref{sec:disc}.

\item Our proposed method can also be used in conjuction with other graph matching 
methods, e.g., as a postprocessing tool. We explore this possibility 
in Section~\ref{sec:exp}, where we show that this approach yields excellent practical performance.
\end{itemize}
\end{itemize}

The paper is organized as follows: Section~\ref{sec:related} presents briefly related work. Sections~\ref{sec:vertexsim} and~\ref{sec:condsim} present 
our proposed methods for the vertex similarity and the graph matching problem respectively. 
Section~\ref{sec:exp} shows an experimental validation and evaluation of our proposed methods. 
Finally, Section~\ref{sec:disc} concludes the paper.

\section{Related Work}
\label{sec:related}
In Sections~\ref{subsec:prelimvertexsim},~\ref{subsec:prelimarchetypes},~\ref{subsec:prelimunmix}
we present work related to our proposed method in Section~\ref{sec:vertexsim}. 
In Sections~\ref{subsec:prelimgraphmatching},~\ref{subsec:prelimcondnum} 
we present work related to our proposed method in Section~\ref{sec:condsim}. 

\subsection{Vertex Similarity}
\label{subsec:prelimvertexsim}

The key idea that appears in different guises 
in the literature related to the vertex similarity problem is the following: 
two vertices are similar if their neighbors are similar. 
The recursive nature of this idea leads to recursive algorithms.
It is worth pointing out that other measures of similarity exist: 
the number of common neighbors, Jaccard's coefficient, Salton's coefficient,    
the Adamic/Acar coefficient \cite{adamic} etc. 
These measures have significant shortcomings. 
For instance, two vertices may be highly similar even if they share no common neighbors \cite{leicht}. 

The algorithm that has influenced and motivated a large and significant part
on vertex similarity is Kleinberg's HITS algorithm \cite{hits}. 
Interestingly, Blondel et al. \cite{blondel2}  generalized HITS 
and provided a general scheme for finding the similarity of two vertices. 
Jeh and Widom proposed the Simrank algorithm \cite{simrank} to compute all-pairs vertex 
similarities in a graph.  Leicht et al. propose another recursive 
measure of similarity closely resembling the centrality measure of Katz \cite{leicht}.
Recursive algorithms are closely connected to spectral graph theory.
Additionally, spectral graph theory through random walks provides the basis for a rich set of similarity measures including commute times
and graph kernel methods \cite{fouss}.
 Recently, non-negative matrix factorization \cite{arora} has been proposed  
in the context of role identification in social networks \cite{henderson2012rolx}.

Vertex similarity has numerous applications such as link recommendation \cite{liben}, schema matching \cite{hector} and privacy attacks \cite{hay}.
Our geometric perspective on the problem of vertex similarity in Section~\ref{sec:vertexsim} has not been considered 
in the literature to the best of our knowledge.

\subsection{Archetypal Analysis} 
\label{subsec:prelimarchetypes} 

The idea of archetypal analysis was born by Breiman during his work on predicting the next-day ozone levels. Breiman
proposed that each day could be quantified as a mixture of ``extreme'' or ``archetypal'' days \cite{cutler}. 
Culter and Breiman introduced archetypal analysis and proposed an alternating minimization procedure \cite{cutler2}. 
Archetypal analysis has numerous applications in various fields including computational biology \cite{huggins} and marketing \cite{whitepaper}.

\subsection{Spectral Unmixing} 
\label{subsec:prelimunmix}

Spectral unmixing is a central problem in spectral imaging. 
Keshava \cite{keshava} surveys existing algorithms for this problem. 
Of special interest to us is the geometric approach, inspired by Craig's seminal work \cite{craig},
where a minimum volume simplex is fitted to the set of points.  

The computational complexity of fitting a minimum volume enclosing simplex depends on the dimensionality $k$. 
Specifically, when $k=2$ there exist efficient algorithms for finding the minimum area enclosing triangle 
\cite{rourke}. When $k=3$, Zhou and Suri give an algorithm with complexity $O(n^4)$ \cite{suri}.
Packer showed that the problem is NP-hard when $k \geq \log{(n)}$ \cite{packer}.


\subsection{Graph Matching} 
\label{subsec:prelimgraphmatching} 

The graph matching problem has attracted a lot of interest  \cite{letters}. 
Umeyama proposed that  instead of trying to find a permutation matrix $P$, i.e., one of the vertices of 
the Birkoff polytope \cite{burk} (it is a polytope whose vertices correspond to permutation matrices,
see the closely related permutahedron in Figure~\ref{fig:fig1}(b)), that minimizes $|| PA_GP^T - A_H||$ where 
$A_G,A_H$ are the adjacency matrix representations of graphs $G,H$ one may relax the problem to finding an orthogonal matrix
$P'$ that minimizes the same objective \cite{umeyama}. 
Other methods relax the constraint of searching for a permutation matrix $P$ to finding a doubly stochastic matrix.

Spectral approaches play a prominent role in matching two shapes,
a key problem in computer vision. 
The problem of shape matching upon preprocessing reduces to 
graph matching \cite{shapiro}. 
Spectral methods take as input two weighted graphs, each representing a shape
and consist typically of three steps \cite{bronstein}. 
We believe that there may be fruitful connections
between our proposed method in Section~\ref{sec:condsim} and the first two steps 
of these methods.
In the first step, the Laplacian embedding of the two shapes is computed.
In the second step,  a permutation matrix $P$ and a sign matrix $S$ which matches the first
$k$ eigenvectors of these shapes to each other are computed. 
Typically, this is done by finding the minimal cost assignment between these vectors. 
The final step is the point registration of these two aligned embeddings, e.g., by using 
the EM-algorithm.  

When the two graphs of interest have different number of vertices, the typical representation is 
the compatibility graph. Using this   representation, the graph matching problem 
can be formulated as an integer quadratic problem (IQP) which can be tackled in various ways. 
Dominating approaches include semidefinite programming \cite{sdp}, spectral approaches \cite{bai},
linear programming relaxations \cite{klau} and the popular graduated assignment method \cite{gold}. 
The latter, relaxes the IQP into a non-convex quadratic program and solves a sequence of convex optimization approximation problems. 
Blondel et al. \cite{blondel2} use a generalization of HITS method \cite{hits} to find graph matchings. 
As we mentioned in Section~\ref{subsec:prelimvertexsim} their method is also applicable to the vertex similarity
problem. Other approaches include Belief Propagation (BP) \cite{gleich} and kernel-based methods \cite{smalter}.

\subsection{Generalized Condition Number} 
\label{subsec:prelimcondnum}

A fundamental problem of linear algebra is solving the linear system of equations $Ax=b$ \cite{golubvanloan}. 
In the case of a preconditioned linear system the corresponding quantity that determines the rate of 
convergence of the solver, e.g., preconditioned conjugate gradient, is the generalized condition number \cite{golub}. 
The definition follows:

\begin{definition}[\cite{gremban}]
Let $A,B$ be two real matrices with the same null space $\field{K}$. 
$\lambda$ is a generalized eigenvalue of the ordered pair of matrices $(A,B)$, also called pencil,
if there exists a vector $x \notin \field{K}$ such that $Ax=\lambda Bx$. Let $\Lambda(A,B)$ the set of generalized
eigenvalues of the pencil $(A,B)$. The generalized condition number $\kappa(A,B)$  is defined as the ratio of the maximum
value $\lambda_{\max}(A,B)$ to the minimum value  $\lambda_{\min}(A,B)$.
\end{definition}

\noindent For every unit norm vector $x$ the following double inequality holds:
\begin{equation}
\lambda_{\min}(A,B) x^T B x \leq x^T A x \leq \lambda_{\max}(A,B) x^T B x.
 \label{eq:eq1}
\end{equation}

\noindent For the special case of interest where the pencil $(L_G,L_H)$ is a pair of Laplacian matrices of two connected
graphs $G,H$ on the same vertex set the condition number is given by the following expression:

$$\kappa(L_G, L_H) = \left( \max_{x^T1 =0} \frac{x^TL_Gx}{x^TL_Hx} \right)  \left( \max_{x^T1 =0} \frac{x^TL_Hx}{x^TL_Gx} \right).$$

Notice that we since $G,H$ are connected their null space is the same and specifically the span of the all-ones vector ${\bf 1}$ \cite{golub}.
Generalized eigenvalue problems of a special form have several important applications in computer science, see for instance \cite{belkin,vision}.

\section{VertexSim: Vertex Similarity via Simplex Fitting} 
\label{sec:vertexsim} 
The main assumption of our proposed method is that each vertex is a ``combination'' of  few ``extreme'' types of vertices.  
This assumption lies conceptually close to archetypal analysis \cite{cutler}.
We formalize mathematically the notions of ``combination'' and ``extreme'' geometrically in the following. 

Our proposed algorithm  is VertexSim shown as Algorithm 1.
The algorithm takes as input the graph $G(V=[n],E)$,  
the dimension of the simplex we wish to fit 
and a parameter $\gamma$ which tunes the sensitivity of the fitting 
algorithm to outliers. We assume that $k \ll n$. 
In the first step, we embed the graph $G$ on the Euclidean space $\field{R}^k$.
Hence, each vertex is mapped to a $k$-dimensional point. 
It is worth emphasizing that typically real-world networks of small and medium size (up to several thousands 
of vertices and edges) have strong geometric structure. As the size grows, the geometry becomes
less apparent, see also Section~\ref{sec:disc}. 
There exist several methods to obtain an informative embedding of the graph. 
The majority of them are spectral \cite{lee}. 
In our experiments we choose the $k$ smallest non-trivial eigenvectors, i.e., the eigenvectors
corresponding the $k$ smallest, non-zero eigenvalues of the normalized Laplacian  \cite{belkin}.

\begin{algorithm} 
 \caption{VertexSim}
 \label{alg1}
 \begin{algorithmic} 
  \REQUIRE Connected, undirected graph $G([n],E)$. Dimension $k$. Parameter $\gamma$. 
  \STATE (1) Embed the graph using the $k$ smallest non trivial eigenvectors of the normalized Laplacian of $G$. 
  \STATE (2) $[K,\{ \theta_i\}_{i \in [n]}] \leftarrow$ Solve Optimization Problem~\ref{prog:robustmix1} using gradient descent.   
  \STATE (3) For every pair of vertices $(i,j)$ compute the Euclidean distance between the mixture 
   coefficient vectors $\theta_i,\theta_j$. 
  \STATE (4) Add points $\{ \theta_i\}_{i \in [n]}$ to a data structure supporting nearest neighbor search queries. 
 \end{algorithmic}
\end{algorithm} 

In the second step,  we learn a simplex, i.e., a set of $k+1$ affinely independent points,
which encloses the cloud of points. The $k+1$ vertices of the simplex are the extreme types of vertices
and each vertex is a convex combination of these types. 
The rationale behind the choice of a simplex is that  each point is expressed
uniquely as a convex combination of the extreme points. 
This allows us to perform a quantitative analysis of vertex similarity
and answer queries 
such as ``which are the three vertices most similar to vertex $v$?'' 
with the use of appropriate data structures \cite{mount}. 
Among all simplexes that fit the cloud of points we favor the one with the smallest volume,
inspired by the seminal work of Craig \cite{craig}.
There exist a wide variety of off-the-shelf algorithms that compute a minimum volume enclosing simplex
and reliable implementations are publicly available. We use the method developed by \cite{tolliver} which solves
the following optimization problem, where $X=\{x_1,\ldots,x_n\} \subseteq \field{R}^k$ is the cloud of points,
$K=[v_0|\ldots|v_k]$ is a simplex in $\field{R}^k$ and $\theta_i \in [0,1]^{k+1}$ for $i=1,\ldots,n$ is the vector of mixture coefficients 
of point $i$:

\begin{eqnarray} 
	\label{prog:robustmix1} ~~~~ ~~~~ \nonumber
	\underset{K,\theta}{min}&:& \sum_{i=1}^s \left |x_i - K \theta_i \right |_p + \gamma \log {\tt vol}(K)\\ \nonumber 
         \\
         \forall \theta_i &:& \theta_i^T{\bf 1} = 1,~\theta_i \succeq 0   
\end{eqnarray}

\noindent The first term in the objective makes the formulation robust to outliers, see \cite{tolliver}. 
We use $|x|_p$ to denote the $p$-norm of vector $x$. We choose $p=1$.
We need to derive the necessary partial derivatives. For completeness we include here the computation. 
Let the simplex be represented by the vertex matrix $K=\left [v_0|...|v_k \right ]$. Then, 
\begin{equation} 
	{\tt vol} (K) = c_k \cdot {\tt det}  \left( \Gamma^T K K^T \Gamma  \right)^{1/2} = c_k \cdot \sqrt{ {\tt det}~Q }
\end{equation} 
where $c_k$ is the volume of the unit simplex defined on $k+1$ points and $\Gamma$ is a fixed vertex-edge incidence matrix such that $\Gamma^T K = \left [v_1-v_0|...|v_k-v_0 \right ]$. 
It follows that 

\begin{eqnarray}
\nonumber
\log {\tt vol}(K) &=&\log c_k + \frac{1}{2} \log {\tt det} Q  \propto \log \prod_{d=1}^k \lambda_d(Q) = \sum_{d=1}^k  \log  \lambda_d(Q). \nonumber 
\end{eqnarray}
Hence, the gradient of $\log {\tt vol}(K)$ is given by 
\begin{eqnarray} 
	\nonumber
\frac{\partial \log {\tt vol(K)}}{\partial K_{ij}} &=& \sum_{d=1}^k \frac{\partial }{\partial K_{ij}} \lambda_d = \sum_{d=1}^k z_d^T(\Gamma^T E_{ij} E_{ij}^T \Gamma) z_d 
\end{eqnarray}
where the eigenvector $z_d$ satisfies the equality $Q z_d = \lambda_d z_d$ and $E_{ij}$ is the indicator matrix for 
the entry $ij$. To minimize the volume, we move the vertices along the paths specified by the negative $\log$ 
gradient of the current simplex volume. 

Finally, in the third and fourth step we compute the similarity between vertices based on 
the set of mixture coefficients. Specifically, we use the Euclidean distance of mixture coefficients
and a data structure which supports nearest neighbor queries \cite{mount} to answer quickly 
typical queries as the ones we have mentioned before. 

It is worth emphasizing that the vertices of the fitted simplex may reveal structure in the network. 
Depending on the network a domain expert can interpret their meaning. In the networks we use 
in Section~\ref{sec:exp} the interpretation is straightforward. 
Because of the special importance of the simplex vertices, we shall refer to them as {\em social network archetypes}.
It is worth noticing that our proposed formulation compared to the $k$-community literature
allows us to mine the graph even when there are no well-shaped clusters, see for instance 
Figures~\ref{fig:fig1}(a),~\ref{fig:stratified} and~\ref{fig:fig4}.

\section{CondSim: Graph Similarity and the Generalized Condition Number} 
\label{sec:condsim} 
\subsection{Theoretical Result and Algorithms}
\label{subsec:condsimalgo} 

Let $G([n],E_G), H([n],E_H)$ be connected graphs on $n$ vertices (labeled for simplicity $\{1,2,..,n\}=[n]$) and $L_G, L_H$ their Laplacian matrix  representation respectively. 
Also, let $\Lambda(L_G,L_H)$ and $\kappa(L_G,L_H)$ be the set of generalized eigenvalues and the generalized condition number of the pencil $(L_G,L_H)$ \cite{golub,golubvanloan}. 
We use $S_n$ to denote the symmetric group, i.e., the group whose elements are all the permutations of the set $[n]$ and whose group operation is the composition of such permutations. 
We denote with $L_{G^(\sigma)}$ where $\sigma \in S_n$ the Laplacian matrix representation of the graph $G$ whose vertex set  has been renamed according to $\sigma$, 
i.e., $v \mapsto \sigma(v)$ for all $v \in [n]$. Our main result is the next theorem. 

\begin{theorem} 
Let $\Omega$ be the state space representing the set of all permutations $\{\sigma: \sigma \in S_n\}$, 
and $f:\Omega \rightarrow \field{R}^+$  be a function defined by $f(\omega) = \kappa(L_G, L_{H^{(\omega)}})$ for all $\omega \in \Omega$.  
Also, fix $\lambda \geq 1$ and define $\pi_{\lambda}(\omega) = \frac{ \lambda^{-f(x)} }{Z(\lambda)}$ 
where $Z(\lambda)=\sum_{x \in \Omega} \lambda^{-f(x)}$ is the normalizing constant  that makes $\pi_{\lambda}$ a probability measure. 
We define a Metropolis chain where we allow transitions between two states if and only if they differ by a transposition
as follows: if $f(\omega_1) < f(\omega_2)$ the Metropolis Chain accepts the transition $\omega_1 \rightarrow \omega_2$ with probability
$\lambda^{ f(\omega_1)-f(\omega_2) }$ otherwise always accept it. 

As $\lambda \rightarrow +\infty$ the stationary distribution $\pi_{\lambda}$ of the Metropolis chain
converges to the uniform distribution over the global minima of $f$. 
Furthermore, if $G \sim H$, i.e., $G,H$ are isomorphic, then  $\pi_{\lambda}$ converges to the uniform distribution over the set of 
isomorphisms $\{\sigma: L_G = L_{H^{(\sigma)}}, \sigma \in S_n \}$.
\label{thrm:thrm1}
\end{theorem}

\begin{proof} 
Recall that  the Laplacian representation of a connected graph  is a symmetric, positive semidefinite matrix and that 
the dimension of its null space is 1 (the all-ones vector ${\bf 1}$).
Consider now the generalized eigenvalue problem $L_G x = \lambda L_H x$.
The pencil $(L_G,L_H)$ is Hermitian semidefinite. Therefore there exists a basis of generalized eigenvectors \cite{golubvanloan}. 
Notice that the all-ones vector ${\bf 1}$ is a generalized eigenvector with corresponding generalized eigenvalue 0.
Let $\Lambda(L_G,L_H)=\{0=\lambda_0< \lambda_1 \leq \ldots \leq \lambda_{n-1}$ be the set of generalized eigenvalues.
Then, $\kappa(L_G,L_H) = \frac{\lambda_{n-1}}{\lambda_1}$\footnote{Notice that despite the fact that our matrices are positive semidefinite,
and not positive definite, this doesn't cause any real problem with respect to defining the generalized condition number since $L_G,L_H$ have
the same null space.}. We prove that $\kappa(L_G,L_H)=1$ if and only $G \sim H$. 

\noindent
\underline{ $\bullet$ {\sc $\kappa(L_G,L_H)=1 \Rightarrow  G \sim H$ :}}

The generalized eigenvalues are $\lambda(L_G,L_H)= ( 0=\lambda_0< 1= \lambda_1 = \ldots = \lambda_{n-1})$. 
Let $( {\bf 1} = u_0, u_1,\ldots,u_{n-1})$ be the corresponding generalized eigenvectors which {\em form a basis}. Define $X = L_G-L_H$. 
Notice that $Xu_i=0$ for all $i=0,..,n-1$. Hence, $X=0$ and therefore $L_G=L_H \rightarrow G \sim H$. 

\noindent
\underline{$\bullet$ {\sc $ G \sim H \Rightarrow \exists \sigma \in S_n \text{~s.t.~} \kappa(L_G,L_{H^{(\sigma)}})=1 $ :}}

Since $G \sim H$ there exists a permutation $\sigma \in S_n$ such that $L_G = L_{H^{(\sigma)}}$. 
Simply, by substituting the eigenvectors $\{u_i\}_{i=0,..,n-1}$ of $L_G$ in $L_G x = \lambda L_{H^{(\sigma)}}x = \lambda L_G x$ 
we obtain that the generalized eigenvalues are $(0,1,1,..,1)$ and the corresponding eigenvectors  $( {\bf 1} = u_0, u_1,\ldots,u_{n-1})$. 
Hence, $\kappa(L_G,L_{H^{(\sigma)}})=1$. 

Now, define $\Omega^* = \{ \omega \in \Omega: f(\omega) = f^* = \min_{x \in \Omega} f(x) \}$. 
Since our chain is a Metropolis chain \cite{peresbook}, its stationary distribution is $\pi_{\lambda}$. 
Therefore,

\begin{align}
 \lim_{\lambda \rightarrow +\infty} \pi_{\lambda}(\omega) &= \lim_{\lambda \rightarrow +\infty}  \frac{ \lambda^{f(\omega)}/\lambda^{f^*} } {|\Omega^*|+\sum_{\omega \in \Omega-\Omega^*} \lambda^{f(\omega)}/\lambda^{f^*} } \\ \nonumber
                                                          &= \frac{ I(\omega \in \Omega^*)}{|\Omega^*|}\nonumber
\end{align}

\noindent where $I(\alpha \in A)$ is an indicator variable equal to 1 if element $\alpha$ belongs to set $A$, otherwise 0. 
If $G,H$ are isomorphic then $f^*=1$ and therefore the above result suggests that the Metropolis Chain converges to the uniform distribution  over the set  $\{\sigma: L_G =L_{H^{(\sigma)}}, \sigma \in S_n \}$.
\end{proof}

It is worth emphasizing that our result is valid even when the two Laplacians are co-spectral. 
For instance Figure~\ref{fig:cospectral} shows two co-spectral graphs with respect to their Laplacians. 
Both Laplacians share the same set of eigenvalues $\{ 0,0.76,2,3,3,5.24 \}$. However, over the space of 
$6!$ permutations the minimum generalized condition number is 6.19.

\begin{figure*}
  \centering
  \subfloat[]{\label{fig:cospectral1}\includegraphics[scale=0.60]{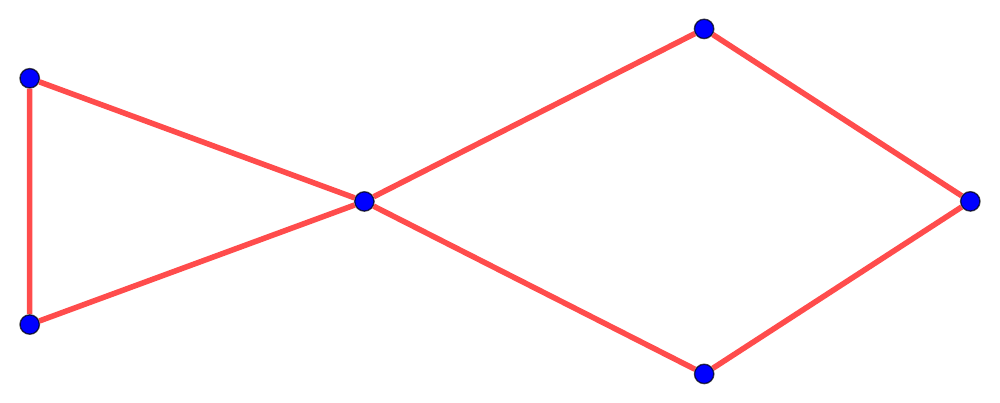}}
  \subfloat[]{\label{fig:cospectral2}\includegraphics[scale=0.60]{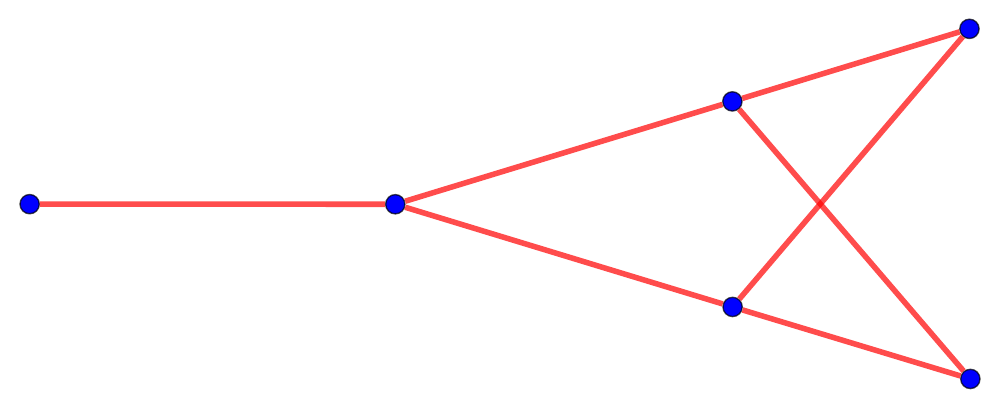}}
  \caption{   Co-spectral Laplacians: The two non-isomorphic graphs have co-spectral Laplacian matrix representation. 
  The minimum generalized condition number over the space of 6! permutations is 6.1852.  }
  \label{fig:cospectral}
\end{figure*}

Our proposed algorithm CondSimGradDescent is shown as Algorithm 2. 
It is a greedy, efficient heuristic. 
The algorithm takes two parameters, the maximum number of iterations $q$ and a 
parameter $\epsilon > 0$ which quantifies the least amount of progress 
required by the algorithm to keep iterating. This may help in avoiding extremely incremental improvements
which do not significantly improve the graph matching but cost a lot computationaly. 
Algorithm 2 performs gradient descent with respect to the generalized condition number using transpositions. 
On the one hand, algorithm 2 tends to be computationally more aggressive than the Metropolis chain 
in the sense that it always moves to a state/permutation which results in a smaller generalized condition number. 
On the other hand the Metropolis chain due to the randomization is likely to avoid local minima. 
Our algorithm returns the permutation which defines the best graph alignment found and the corresponding condition number. 

The complexity  of our algorithm 2  depends on the choice of algorithm  that solves the generalized eigenvalue problem. 
Specifically, let $f(L_G,L_H)$ be the corresponding running time as a function of the two Laplacians. Also let $q$ abbreviate the maximum number of iterations MAXITER. 
Then the total running time is upper bounded by $O(q n^2 f(L_G,L_H))$ since we perform $q$ steps and at each step we compute the generalized condition number
for the ${n \choose 2}$ possible transpositions. 
In our experiments we use the algorithm of Golub and Ye \cite{golub}. The speed of convergence is given in Lemma 1, p. 8 \cite{golub}. 
For our purposes, since we set the number of iterations (which are matrix-vector multiplications) of the Golub-Ye algorithm
to a constant we may assume that the running time that computing the smallest non-trivial and the largest generalized eigenvalue of the pencil $(L_G,L_H)$
is linear in the total number of edges $|E_G|+|E_H|=O(m)$ where $m = \max{ (|E_G|,|E_H|)}$. 
It is worth noticing that using a series of transpositions we can reach any permutation from any starting 
permutation. 
If $m$ is large, e.g., $m \gg n\log{n}$, one can use the developed theory of spectral sparsifiers  
to speed up the  generalized condition number computations. Specifically, one may perform first the   Spielman-Srivastava
sparsification \cite{srivastava} on both Laplacians $L_G,L_H$, obtain spectrally equivalent matrices $\tilde{L}_G,\tilde{L}_H$
and apply our algorithm on the latter Laplacians.

\begin{algorithm} 
 \caption{CondSimGradDescent}
 \label{alg3} 
 \begin{algorithmic} 
  \REQUIRE $L_G,L_H$ the Laplacian matrix representation of $G,H$ respectively.  $q$ (Maximum number of iterations). $\epsilon > 0$ (Tolerance).
  \COMMENT{$\sigma$ initialized to the identity permutation}
  \STATE $\sigma \leftarrow (1,2,..,n)$ 
  \STATE $i \leftarrow 0$ 
   \WHILE{$i \leq q$  }
    \STATE $i \leftarrow i+1$
    \STATE $\sigma^* \leftarrow \arg\max_{\sigma' \in S'}  \kappa(L_G,L_{H^{(\sigma)}})-\kappa(L_G,L_{H^{(\sigma')}}) $ 
     where $S'$ is  the set of all permutations which differ from $\sigma$. a single transposition.
     \IF{ $\kappa(L_G,L_{H^{(\sigma)}})-\kappa(L_G,L_{H^{(\sigma')}}) > \epsilon$ }
          \STATE $\sigma \leftarrow \sigma^*$
          \STATE CN$\leftarrow \kappa(L_G,L_{H^{(\sigma)}})$ 
     \ELSE 
          \STATE {\it break}
     \ENDIF
     \IF{CN $=1$}
       \STATE {\it break}
     \ENDIF 
  \ENDWHILE
  \STATE Return ($\sigma$, CN)
 \end{algorithmic}
\end{algorithm}

\subsection{Further Insight: Theory of Support Trees}
\label{subsec:insight} 

We proved in Theorem~\ref{thrm:thrm1} that when the generalized condition number is 1, then 
indeed $G,H$ can be perfectly matched, i.e., $G,H$ are isomorphic. To complete the justification of our rationale
behind the choice of our similarity measure    we need to explain why does a value close to 1 
imply a good graph alignment.  The answer lies in the theory of support preconditioners \cite{gremban,support}. 
In the following, let  $A,B$ be Laplacian matrices. 

\begin{definition}[Support] 
The support $\sigma(A,B)$ of matrix $B$ for $A$ is the greatest lower bound over all $\tau$ such that $\tau B -A$ is positive
semidefinite, i.e., $\sigma(A,B) = \lim \text{inf} \{ \tau: \tau B - A \succeq 0\}$. 
\end{definition}

\begin{definition}[Congestion \& Dilation] 
An embedding of $H$ into $G$ is a mapping of vertices of $H$ onto vertices of $G$, and edges of $H$
onto paths in $G$. The dilation $d(G,H)$ of the embedding is the length of the longest path in $G$ onto which an edge of $H$ is
mapped. The congestion $g_e(G,H)$ of an edge $e$ in G is the number of paths of the embedding that contain $e$. The congestion $g(G,H)$ 
of the embedding is the maximum congestion of the edges in G. 
\end{definition}

The following facts have been proved in Gremban's Ph.D. thesis \cite{gremban}: 
(a) The support number $\sigma(A,B)$ is bounded above by the maximum product of dilation and congestion  over all embedding of $A$ into $B$. 
(b) $\kappa(A,B) \leq \sigma(A,B) \sigma(B,A)$, Lemma 4.8 \cite{gremban}.
Fact (b), in combination with fact (a), shows that the generalized condition number is closely related to 
the goodness of two embeddings, i.e., of $H$ into $G$ and vice versa. When both the dilation and the 
congestion of the embeddings --or in our terminology of the alignments-- are small then the generalized condition number is small.


\section{Experiments}
\label{sec:exp}
\begin{table}[ht]
\centering 
\begin{tabular}{|l|r|r|} \hline
Name (Abbr.)                          & Nodes (n) & Edges (m) \\ \hline
\textcolor{red}{$\odot$} Netscience   & 1589      & 2742     \\ \hline
\textcolor{red}{$\odot$} Football     & 115       & 613     \\ \hline  
\textcolor{red}{$\odot$} Political Books & 105 & 441     \\ \hline
\textcolor{cyan}{$\star$} Erd\"{o}s '72   & 5488 & 7085     \\ \hline
\textcolor{cyan}{$\star$} Erd\"{o}s '82    & 5822 & 7375     \\ \hline
\textcolor{cyan}{$\star$} Erd\"{o}s '02    & 6927 & 8472     \\ \hline
\textcolor{cyan}{$\star$} Roget Thesaurus     & 1022 & 3648     \\ \hline
\end{tabular}
\caption{Datasets}
\label{tab:datasets}
\end{table}

In Sections~\ref{subsec:datasets},~\ref{subsec:expset} we describe the datasets we used in our experiments and  the experimental setup.
In Sections~\ref{subsec:simplexsimexp},~\ref{subsec:condsimexp} we provide an experimental evaluation of our proposed methods respectively.

\subsection{Datasets} 
\label{subsec:datasets} 

\begin{figure} 
\centering
\includegraphics[width=0.5\textwidth]{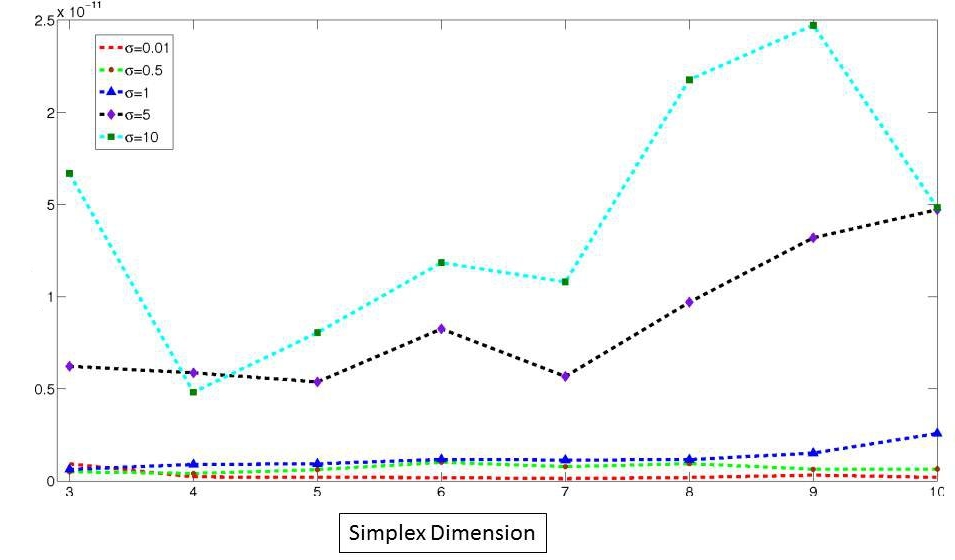} 
\caption{Performance of simplex fitting on 1000 points drawn uniformly at random from a 
randomly generated $k$-simplex perturbed by Gaussian noise $N(0,\sigma^2)$. Figure plots the sum of Euclidean distances of the  $k+1$ reconstructed
simplex vertices from the $k+1$ true vertices as a function of the dimensionality $k$ of the simplex
for five  different standard deviations $\sigma=0.01,0.5,1,5,10$. Notice that the simplex fitting method (essentially)
perfectly recovers the true simplex in all cases.} 
\label{fig:synthsimp} 
\end{figure}

Table~\ref{tab:datasets} summarizes the real-world datasets we used for our experiments. 
Whenever neccessary, graphs are made undirected, unweighted and self-edges were removed. 
Datasets annotated with \textcolor{red}{$\odot$} and \textcolor{cyan}{$\star$}  are available online from 
\cite{misc1,misc2} respectively. 
We pick small and medium sized networks deliberately since the geometric
structure in such networks is striking. 

We also generate several synthetic datasets. 
Specifically,  for Section~\ref{subsubsec:synthsimplexsim} we generate a cloud of points where each point is chosen uniformly at random 
from a random $k$-simplex (see Appendix \cite{peresbook}) and a random stratified social network,
see Section IIIA of \cite{leicht}. 
Notice that the first type of synthetic data involves no graph and its goal is to test the 
goodness of the simplex fitting method.
Stratified networks model the phenomenon according to which  individuals  make connections with individuals similar to them
 with respect to some criterion, e.g., income, age. 
For each vertex we pick an age from 1 to 10, chosen uniformly at random. 
Two vertices with age $i$ and $j$ respectively are connected with probability $p_0 e^{-\alpha \Delta t}$
where $\Delta t = |i-j|$. The parameters are set to $\alpha=0.8$ and $p_0=0.1$. 
Finally, for Section~\ref{subsec:condsimexp} we generate random graphs of two types. 
Erd\"{o}s-R\'{e}nyi-Gilbert  graphs \cite{bela} and R-MAT graphs \cite{rmat}. 
For the former we use $p=0.5$ and for the latter the parameters
are set to $[a = 0.55, b = 0.1; c = 0.1, d = 0.25]$.

\begin{figure} 
\centering
\includegraphics[width=0.5\textwidth]{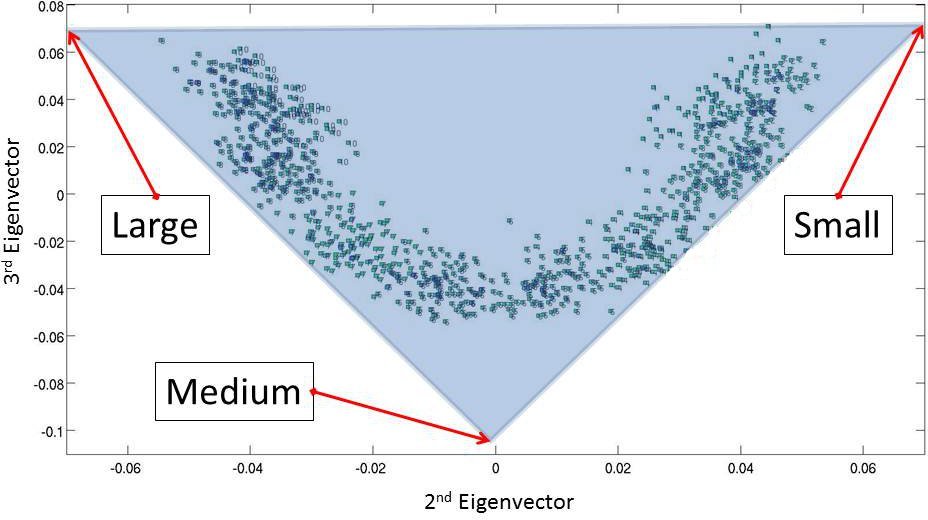} 
\caption{2-simplex fitted on a random stratified network. VertexSim correctly assigns higher similarity
values to vertices of the same age. The three vertices of the fitted 2-simplex conceptually represent
the concepts 'senior/large age' (8-10), 'middle-aged/medium age'(4-7) and 'young/small age' (1-3).}
\label{fig:stratified} 
\end{figure} 

\subsection{Experimental Setup and Implementation Details}
\label{subsec:expset}

The experiments were performed on a single machine, with Intel Xeon CPU
at 2.83 GHz, 6144KB cache size and and 50GB of main memory. 
Our algorithms are implemented in MATLAB.   
The results we show are obtained for setting the parameter $\gamma$ was set to 1 and the dimensionality $k$ of the embedding equal to 2.
Clearly our method is valuable when $k$ is larger than 3 where visualization is impossible (see also 
Section~\ref{subsec:prelimunmix} for a discussion of the computational complexity as a function of $k$). Here we report results for $k=2$
for visualization purposes.  It is worth pointing out two more facts concerning our experimental section: 
VertexSim in our experiments was not affected by the value of the parameter $\gamma$  since there were no outliers
in any of the embeddings and secondly we experimented with higher values of $k$ (from 3 to 5) obtaining interpretable results.  

  The wall-clock times we report use the Golub-Ye  algorithm \cite{golub} as a subroutine to compute condition numbers. 
In order to use this algorithm which is designed for positive definite pencils we shift slightly the spectrum of the Laplacians,
namely we set $L' = L+ \epsilon \frac{{\mathbf 1} {\mathbf 1}^T }{n}$ where $\epsilon$ is a small positive constant. 
This is a natural ``trick'' to compute the generalized condition number.
We use the default settings of the {\it eigifp} software which is available online at \url{http://www.ms.uky.edu/~qye/software.html}. 
It is worth mentioning that given that our graphs are small- and medium-sized, we checked the quality 
of this ``trick''. We observed  that when we set $\epsilon=0.01$, we obtain essentially accurate condition numbers. 
For instance, assume we permute the set of vertices of the Football network according to a randomly 
generated permutation. 
We compute the generalized condition number using the shifting trick and the Golub-Ye algorithm
and exactly by computing the eigenvalues of $(L_b)^\dagger L_a$. In the former case we obtain 
52.592 and in the latter 52.591.  This is a representative example of what we observe in practice,
which also explains why this shifting heuristic is frequently used, see Section 6 in \cite{sunicml}.

The parameters of CondSimGradDesc were set to $q=200$, $\epsilon=0$ for all experiments in Section~\ref{subsec:condsimexp}.

A final remark with respect to the experiments of CondSim in Section~\ref{subsec:condsimexp}:
it is a well known fact that a permutation can be decomposed in cycles and that a random permutation has $O(\log{n})$ cycles  in expectation \cite{wilf}.  
Therefore, if we generate only permutations chosen uniformly at random we are restricting ourselves with high probability
to permutations which share a common structure. To avoid a potential artifact in our experimental results, we generate permutations
with a  different  number of cycles. We use a simple recursive algorithm \cite{wilf} to generate a permutation with $k$ cycles uniformly at random  
in our experiments, see p.33 \cite{wilf}.
Finally, we use third-party software and specifically the code of \cite{mount,gleich} and Jeremy Kepner's R-MAT code implementation.

\begin{figure*}
  \centering
  \subfloat[Football]{\label{fig:fig4a}\includegraphics[scale=0.22]{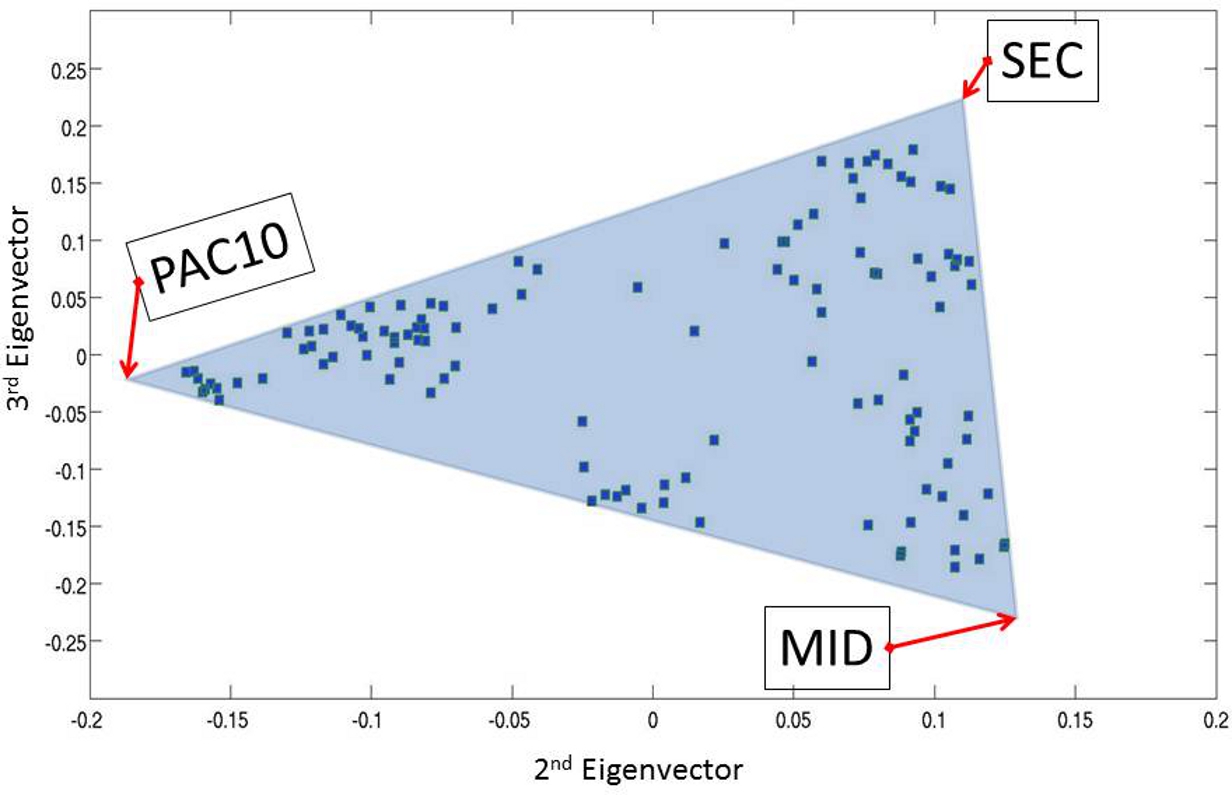}}
  \subfloat[Political books]{\label{fig:fig4b}\includegraphics[scale=0.30]{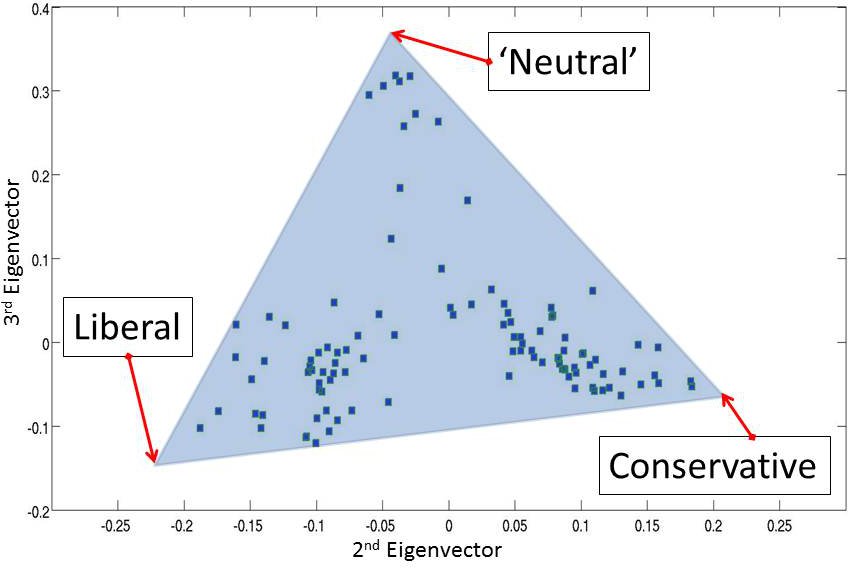}}
  \caption{ Minimum area 2-simplexes for the (a) Football network (b) Political books network. In both 
   cases VertexSim provides significant mining capabilities for extracting pairs of highly similar vertices 
   and concepts. For more see Section~\ref{subsubsec:realsimplexsim}.}
  \label{fig:fig4}
\end{figure*}

\subsection{VertexSim at Work}
\label{subsec:simplexsimexp}

\subsubsection{Synthetic Data} 
\label{subsubsec:synthsimplexsim}

We validate the VertexSim algorithm in two ways:  first we verify that it can successfully recover the simplex $\mathcal{S}$
and hence the mixture coefficients of data points sampled uniformly at random from $\mathcal{S}$
and secondly we evaluate its performance on a stratified network. 
Figure~\ref{fig:synthsimp} shows the performance of our fitting method
as a function of the simplex dimension (x-axis) for five different standard deviations $\sigma=0.01,0.5,1,5,10$ (5 lines) 
for a randomly generated $k$-simplex.  The quality of the performance (y-axis) is quantified as the sum $ \sum_{i=1}^{k+1} || v_i - \tilde{v}_i|| $
where $\tilde{v}_i$ is the reconstructed vertex of the $k$-simplex. 
The performance is excellent as Fig.~\ref{fig:synthsimp} shows. The average running time for four executions is 0.0071 and the variance $1.4 \times 10^{-6}$. 
It is worth mentioning that we also tried the Chan et al. algorithm \cite{tsang1} obtaining exactly the same simplex.

Figure~\ref{fig:stratified} shows the performance of VertexSim for a stratified network with 
$\alpha=0.8$ and $p_0=0.1$ and ages ranging from 1 to 10, picked uniformly at random
for every vertex. Specifically Fig.~\ref{fig:stratified} shows the fitted 2-simplex. 
Upon performing step 3 of Algorithm 1, it becomes apparent that pairs of vertices with the 
same age are significantly more similar than vertices with different ages, as a good 
vertex similarity algorithm should have as its output. Furthermore the three vertices 
of the 2-simplex correspond to the three concepts 'senior/large age', 'middle-aged/medium age'
and 'young/small age'.

\subsubsection{Real-world Data}
\label{subsubsec:realsimplexsim}

Figure~\ref{fig:fig4}(a) shows the minimum area fitted 2-simplex for the American football college network,
whose vertices correspond to teams and edges to games among them. 
According to \cite{girvan} the teams are divided into conferences containing around 8–12 teams each 
and the frequency of games between members of the same conference is higher than between members
of different conferences.  The three vertices of the fitted simplex correspond to three conferences
PAC 10, SEC and MID. Furthermore, VertexSim using the fitted mixture coefficients 
assigns higher similarity to vertices of the same conference. 
The fitting algorithm took 7.5 seconds to find the simplex.
Similar remarks hold for Figure~\ref{fig:fig4}(b) which shows the 
minimum area fitted 2-simplex for the political books network
whose vertices represent books and edges copurchasing by the same buyer. 
The three vertices of the simplex correspond to liberal, convervative
and 'neutral' books.  For both datasets, the vectors of mixture coefficients  provide us a 
novel way to determine vertex similarities in an interpretable way. 
The fitting algorithm needs 6 seconds to compute the simplex. 
We obtain highly interpretable results for other datasets as well. 
Indicately we report few highly similar pairs of vertices according to VertexSim:  (musician, poetry),
(melody, poetry), (voice, hearing) from the Roget Thesaurus network and 
(Vojtech R\"{o}dl, Noga Alon), (Joel Spencer, Jan\`{o}s Pach) from the 
Erd\"{o}s collaboration network, 1972.  

\begin{table} 
\centering 
\begin{tabular}{ |c |c |c |c |c |c |c|} \hline
\cellcolor[gray]{0.9}      &         \multicolumn{3}{|c|}{Erd\"{o}s-Renyi}           &     \multicolumn{3}{|c|}{R-MAT}               \\ \hline
$|V|$                       &         8        &      16             & 32             &         8         &       16        & 32                 \\ \hline
CondSimGradDescent                   &         6/7      &      5/15           & 7/31           &         5/7       &       5/15      & 11/31                 \\ \hline
Belief Propagation (BP)\cite{gleich}                 &         0/7      &      0/15           & 0/31           &         0/7       &       0/15      & 0/31                 \\ \hline
\end{tabular}
\caption{Results of our method versus the Belief Propagation method of \cite{gleich} on various random networks for permutations 
whose number of cycles ranges from 1 to $|V|-1$. The fractions indicate how many times did an algorithm find a permutation
which makes the original graph and its permuted version exactly the same. } 
\label{tab:results} 
\end{table}

\subsection{CondSimGradDescent at Work}
\label{subsec:condsimexp}

\subsubsection{Synthetic Data} 
\label{subsubsec:condsimsynthetic} 

We compare CondSimGradDescent with the belief propagation method of \cite{gleich} for few synthetic datasets
in the following way. We generate a random graph of $n$ vertices and a permutation with $k$ cycles, where 
$k$ ranges from 1 to $n-1$. We do not consider the identity permutation with $n$ cycles. 
We permute the graph according to the random permutation and see whether the graph matching 
methods can align perfectly the original graph and its permuted version. We use two types of random graphs, 
namely binomial random graphs \cite{bela}  and R-MAT graphs with 8, 16, 32 vertices. 
Table~\ref{tab:results} shows the results. As we see, CondSimGradDescent outperforms significantly BP \cite{gleich}. 
It is worth pointing out again that this is a validation test and that fast graph isomorphism
tests exist \cite{bela}.
An interesting trend we observe is that the fewer the cycles of the permutation, 
the easier CondSimGradDescent gets trapped into local minima. 
On the positive side, when the number of cycles is small CondSimGradDescent typically finds an optimal alignment efficiently.

\subsubsection{CondSimGradDescent as a Post-Processing Tool} 
\label{subsubsec:condsimreal} 

Due to the computational cost of CondSimGradDescent, a realistic use of it in large networks is as a post-processing tool. 
We describe a typical use of CondSimGradDescent as a post-processing tool which 
significantly improves the graph alignment in combination with the Belief-Propagation
based method of \cite{gleich}. 
We perform the following experiment: we consider the {\it Football} network. 
We generate a permutation $\sigma$ uniformly at random and permute the labels of $G$ accordingly. 
The number of fixed points of the permutation we obtained is 0. 
We apply the function {\it netalignbp()} which  is open-sourced \cite{gleich}. The alignment produced by \cite{gleich} 
has recognized correctly 50 out of the 115 correct vertex to vertex assignments. The generalized condition number equals 29.4. 
Applying the CondSimGradDesc method to the alignment obtained from Belief Propagation, we obtain a generalized condition number of value 4.16
resulting in 74 correct assignments. Each iteration of CondSimGradDesc lasts in average 90 seconds with 
standard deviation equal to 2 seconds over the 200 iterations.
  It is worth outlining that the Belief-Propagation based method of \cite{gleich} 
is designed for the setting where there exists a reasonable guess for the graph alignment. 
One should conclude only that our proposed method is useful as a postprocessing tool,
and not draw any negative conclusions on the performance of \cite{gleich}.

\section{Conclusion}
\label{sec:disc} 
{\bf Summary:} In this work we contribute to the important problem of quantifying 
vertex similarity in networks by introducing novel approaches to
the problems of vertex similarity within and across
two graphs with the same number of vertices respectively. 
We observed an excellent performance of our algorithms both 
on synthetic and real-world networks. This verifies empirically that both 
the proposed conceptual approaches as well as the algorithmic solutions 
are valuable. 

{\bf Discussion \& Open Problems:} Concerning  our first 
algorithm VertexSim, two natural questions arise: Is there always geometric structure in social networks?  Can we fit other geometric objects 
such as simplicial complexes to capture more complex geometric structure? 
Concerning the first question, Leskovec et al. \cite{jure} have studied extensively properties of large scale networks and it appears that there 
exists strong geometric structure in small and medium sized networks like the ones we studied in Section~\ref{sec:exp} 
but the structure typically decays as the size of the network grows. 
The answer to the second question is an interesting research problem.

Concerning our second algorithm CondSimGradDescent, an interesting research direction is to extend it to cases where 
the two graphs have a different number of vertices. 
The theory of Steiner tree preconditioners \cite{koutis} is a promising 
approach.  Also, understanding the performance of CondSimGradDescent in the isomorphism setting 
is another interesting problem. 
Finally, understanding its performance in simple graphs, e.g., trees, 
remains open.

 \section*{Acknowledgements} 

Research supported by NSF Grant No. CCF-1013110. 
I would like to thank Ioannis Koutis, Gary Miller for several discussions 
on the theory of support tree preconditioners and Jure Leskovec for 
discussions concerning the geometry of social networks. 
Also, I would like to thank the reviewers for their thorough comments.



\end{document}